\documentclass[reprint, showpacs, superscriptaddress]{revtex4-1}

\usepackage{amsthm, CJK, graphicx, hyperref}

\newtheorem{corollary}{Corollary}
\newtheorem{lemma}{Lemma}
\newtheorem{proposition}{Proposition}
\newtheorem{theorem}{Theorem}

\theoremstyle{definition}
\newtheorem{definition}{Definition}

\theoremstyle{remark}
\newtheorem*{remark}{Remark}

\begin{document}

\begin{CJK*}{UTF8}{}

\title{Quantum circuit complexity of one-dimensional topological phases}

\CJKfamily{gbsn}

\author{Yichen Huang (黄溢辰)}

\email{yichenhuang@berkeley.edu}

\affiliation{Department of Physics, University of California, Berkeley, Berkeley, California 94720, USA}

\author{Xie Chen}

\affiliation{Department of Physics, University of California, Berkeley, Berkeley, California 94720, USA}
\affiliation{Department of Physics and Institute for Quantum Information and Matter, California Institute of Technology, Pasadena, California 91125, USA}

\date{\today}

\begin{abstract}

Topological quantum states cannot be created from product states with local quantum circuits of constant depth and are in this sense more entangled than topologically trivial states, but how entangled are they? Here we quantify the entanglement in one-dimensional topological states by showing that local quantum circuits of linear depth are necessary to generate them from product states. We establish this linear lower bound for both bosonic and fermionic one-dimensional topological phases and use symmetric circuits for phases with symmetry. We also show that the linear lower bound can be saturated by explicitly constructing circuits generating these topological states. The same results hold for local quantum circuits connecting topological states in different phases.

\end{abstract}


\maketitle

\end{CJK*}

\section{Introduction}

Many-body entanglement is essential to the existence of topological order in strongly correlated systems. While ground states in topologically trivial phases can take a simple product form, ground states in topological phases are always entangled. Of course, ground states in topologically trivial phases can be entangled, too. It is then natural to ask what is the essential difference between the entanglement patterns that give rise to topologically trivial and nontrivial states.

Besides topological entanglement entropy \cite{KP06, LW06} and the entanglement spectrum \cite{LH08}, which partially capture the topological properties of the system, quantum circuits \cite{NC00} provide a powerful tool for characterizing the entanglement patterns of topological states. Intuitively, one would expect that states with more complicated entanglement patterns require larger circuits to generate from product states. Also, small circuits would suffice to connect ground states in the same phase as their entanglement patterns are similar, while large circuits are necessary to map states from one phase to another.

Indeed, in gapped quantum many-body systems it has been shown that two ground states are in the same topological phase if and only if they can be mapped to each other with a local quantum circuit of constant depth, i.e., a constant (in the system size) number of layers of nonoverlapping local unitaries \cite{CGW10}. States with nontrivial intrinsic topological order are thus said to be long-range entangled in the sense that they cannot be created from product states with circuits of constant depth. Circuits of constant depth can generate symmetry protected topological (SPT) states from product states but only if the symmetry is broken. If only symmetric unitaries are allowed, the circuit depth has to grow with the system size.

Therefore, topological states are in this sense more entangled than topologically trivial states, but how entangled are they? In particular, we ask, what is the quantum circuit complexity of generating topological states from product states, i.e., how does the circuit depth scale with the system size? In two and higher dimensions, it has been shown that circuits of linear (in the diameter of the system) depth are necessary to generate states with topological degeneracy \cite{BHV06}. One might expect that topological states without topological degeneracy are less entangled and can be created with circuits of sublinear depth. However, we show that this is not the case, at least in one dimension (1D).

We demonstrate that, \emph{to generate 1D gapped (symmetry protected) topological states from product states, the depth of the (symmetric) local quantum circuits has to grow linearly with the system size.} The Majorana chain \cite{Kit01} provides an example of a topological state without topological degeneracy, and we show that local fermionic circuits of linear depth are necessary for its creation. For all 1D SPT states, we show that linear depth is required as long as the symmetry is preserved. In particular, we prove that the nonlocal (string) order parameters \cite{HPCS12, PT12} distinguishing different SPT phases remain invariant under symmetric circuits of sublinear depth. Furthermore, we explicitly construct circuits of linear depth that generate 1D topological states. These results suggest the dichotomous picture that ground states of gapped local Hamiltonians are connected by local quantum circuits of either constant or linear depth, depending on whether they are in the same phase or not.

The paper is organized as follows. Section \ref{sec2} reviews the basic notion of gapped quantum phases and how 1D topological phases are classified with local quantum circuits (Appendixes \ref{rig} and \ref{asec2}). Then we study the quantum circuit complexity of prototypical examples of 1D topological phases: the Majorana chain in fermionic systems (Sec. \ref{sec3}) and the Haldane chain with $Z_2\times Z_2$ on-site symmetry in bosonic (spin) systems (Sec. \ref{sec4} and Appendix \ref{full}). We explicitly construct circuits of linear depth that generate these topological states from product states (Propositions \ref{prop1} and \ref{prop3}) and show that linear depth is a lower bound (Propositions \ref{prop2} and \ref{prop4}). For the Majorana chain, the circuit is composed of fermionic local unitaries; for the Haldane chain with symmetry, the circuit is composed of symmetric local unitaries. Appendixes \ref{all_ld} and \ref{all_llb} establish the same results for all 1D topological phases in a similar but more complicated way. Section \ref{sec5} concludes with the implications of our results.

\section{Preliminaries} \label{sec2}

We first review the basic notions of gapped quantum phases and local quantum circuits.

\begin{definition} [gapped quantum phase] \label{def1}
Two gapped local Hamiltonians $H_0$ and $H_1$ are in the same phase if and only if there exists a smooth path of gapped local Hamiltonians $H(t)$ with $0\le t\le1$ such that $H(0)=H_0$ and $H(1)=H_1$. Correspondingly, their ground states are said to be in the same phase.
\end{definition}

Indeed, gapped phases can be defined purely in terms of the ground states, without referring to their Hamiltonians at all. To do this, we need local quantum circuits.

\begin{definition} [local quantum circuit]
A local quantum circuit $C$ of depth $m$ has a layered structure of local unitary quantum gates,
\begin{equation} \label{eq1}
C=\prod_{i_m}C_{i_m}^{(m)}\prod_{i_{m-1}}C_{i_{m-1}}^{(m-1)}\cdots\prod_{i_1}C_{i_1}^{(1)},
\end{equation}
where in each layer $1\le k\le m$ the supports of the local unitaries $C_{i_k}^{(k)}$'s are pairwise nonintersecting.
\end{definition}

\begin{theorem} [informal statement] \label{thm1}
Gapped ground states in the same phase are connected by local quantum circuits of constant depth (up to some reasonably small error).
\end{theorem}

\begin{remark}
See Theorem \ref{Adi} in Appendix \ref{rig} for the formal statement of Theorem \ref{thm1}.
\end{remark}

Theorem \ref{thm1} was discussed in Ref. \cite{CGW10} using quasiadiabatic continuation \cite{HW05, BHM10} and the Lieb-Robinson bound \cite{LR72, NS06, HK06}. Gapped phases can also be defined in the presence of symmetry.

\begin{definition} [symmetry protected topological (SPT) phase]
In the absence of symmetry breaking, two symmetric gapped local Hamiltonians $H_0$ and $H_1$ are in the same SPT phase if and only if there exists a smooth path of symmetric gapped local Hamiltonians $H(t)$ with $0\le t\le1$ such that $H(0)=H_0$ and $H(1)=H_1$.
\end{definition}

SPT phases can also be defined purely in terms of the symmetric ground states.

\begin{definition} [symmetric local quantum circuit]
A local quantum circuit $C$ (\ref{eq1}) is symmetric if each quantum gate $C_{i_k}^{(k)}$ is symmetric.
\end{definition}

\begin{corollary} [informal statement] \label{cor1}
Symmetric gapped ground states in the same SPT phase are connected by symmetric local quantum circuits of constant depth (up to some reasonably small error).
\end{corollary}

\begin{remark}
See Corollary \ref{coro} in Appendix \ref{rig} for the formal statement of Corollary \ref{cor1}.
\end{remark}

Based on Theorem \ref{thm1} and Corollary \ref{cor1}, 1D gapped phases have been classified \cite{PTBO10, CGW11, CGW11a, TPB11, FK11, SPC11}. It was found that there is no topological phase in 1D bosonic (spin) systems without symmetry. In 1D fermionic systems without extra symmetry (beyond fermion parity which is always preserved), there is one and only one topological phase: the Majorana chain with Majorana edge modes \cite{Kit01}. In 1D systems with (extra) symmetry, there can be SPT phases with degenerate edge states carrying projective representations of the symmetry group. See Appendix \ref{asec2} for the classification of 1D SPT phases.

Since (symmetry protected) topological states cannot be mapped to topologically trivial states (including product states) with (symmetric) local quantum circuits of constant depth, we ask, what circuit depth is necessary to do this mapping? We show that linear depth is necessary by proving the invariance of the nonlocal (string) order parameters \cite{BV13, HPCS12, PT12} distinguishing different (symmetry protected) topological phases under (symmetric) circuits of sublinear depth.

\begin{theorem} \label{thm2}
Suppose $|\psi\rangle$ and $C|\psi\rangle$ are two gapped ground states in 1D systems (with symmetry), where $C$ is a (symmetric) local quantum circuit of sublinear depth. Then $|\psi\rangle$ and $C|\psi\rangle$ are in the same (symmetry protected) topological phase.
\end{theorem}

\section{Majorana chain} \label{sec3}

In the absence of (extra) symmetry (beyond fermion parity), the Majorana chain with Majorana edge modes \cite{Kit01} is the only 1D topological order. We now study the Majorana chain by considering the fermionic model
\begin{eqnarray} \label{maj}
H&=&\sum_{j=1}^{N-1}(a_j-a_j^\dag)(a_{j+1}+a_{j+1}^\dag)+\mu\sum_{j=1}^N(2a_j^\dag a_j-1)\nonumber\\
&&-(a_N-a_N^\dag)(a_1+a_1^\dag)
\end{eqnarray}
with antiperiodic boundary conditions in the symmetry sector of even fermion parity, where $a_j$ and $a_j^\dag$ are the fermion annihilation and creation operators at the site $j$. The model (\ref{maj}) is in the topologically trivial and nontrivial phases for $\mu>1$ and $0\leq \mu<1$, respectively. We show that two ground states in different phases can be connected by a local fermionic circuit of linear depth and that linear depth is a lower bound.

\begin{figure}
\includegraphics[width=\linewidth]{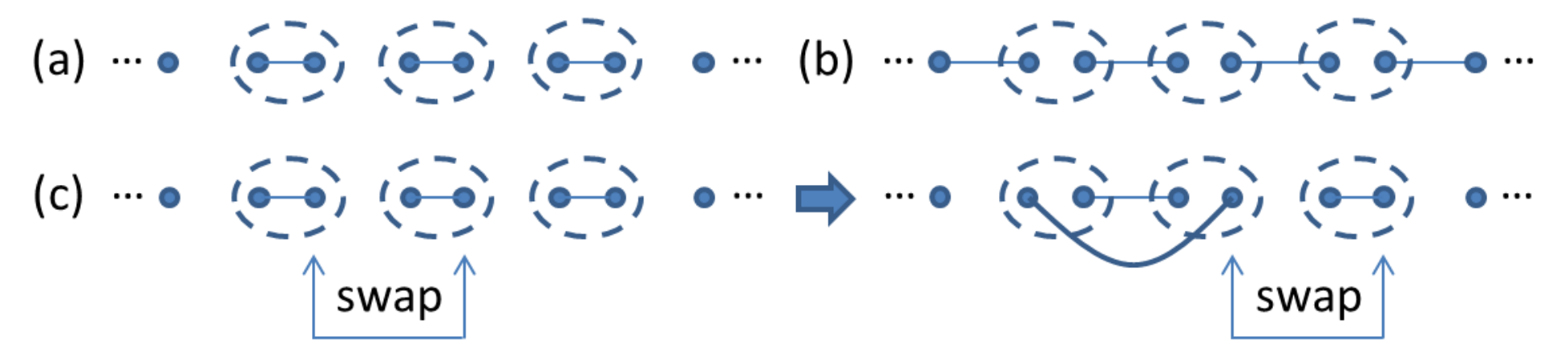}
\caption{(Color online) The renormalization group (RG) fixed-point states \cite{VCL+05, CGW11} in the (a) trivial and (b) nontrivial fermionic (Majorana chain) or SPT (e.g., Haldane chain) phases. For states in fermionic phases, each dot represents a Majorana mode and connected pairs form fermionic modes which are vacant or occupied. For states in SPT phases, each dot carries a projective representation of the symmetry group and connected pairs form symmetric singlets. (c) The states in (a) and (b) can be exactly mapped to each other with a linear-depth $2$-local quantum circuit composed of \textsc{swap} gates.} \label{swap}
\end{figure}

\begin{proposition} \label{prop1}
Suppose $|\psi_0\rangle$ and $|\psi_1\rangle$ are two gapped ground states in the topologically trivial and nontrivial phases in 1D fermionic systems, respectively. Given an arbitrarily small constant $\epsilon$, there exist $|\psi'_0\rangle,|\psi'_1\rangle$ and a local fermionic circuit $C$ of linear depth such that $|\psi'_1\rangle=C|\psi'_0\rangle$ and
\begin{equation}
|\langle\psi_k|P|\psi_k\rangle-\langle\psi'_k|P|\psi'_k\rangle|\le\epsilon~(k=0,1)
\end{equation}
for any local operator $P$ with bounded norm.
\end{proposition}

\begin{proof}
Define two Majorana operators at each site:
\begin{equation}
c_{2j-1}=a_j+a_j^\dag,~c_{2j}=(a_j-a_j^\dag)/i.
\end{equation}
At $\mu=+\infty$, $H=i\mu\sum_{j=1}^Nc_{2j-1}c_{2j}$ is in the trivial phase, and its ground state $|\phi_0\rangle$ is the tensor product of the vacuum states of the modes $a_j = (c_{2j-1}+ic_{2j})/2$. At $\mu=0$, $H=i\sum_{j=1}^{N-1}c_{2j}c_{2j+1}-ic_{2N}c_1$ is in the nontrivial phase, and its ground state $|\phi_1\rangle$ is the tensor product of the vacuum (or occupied) states of the fermionic modes $b_j = (c_{2j}+ic_{2j+1})/2$. Figures \ref{swap}(a) and \ref{swap}(b) illustrate the structures of $|\phi_0\rangle$ and $|\phi_1\rangle$, which are the RG fixed-point states in the topologically trivial and nontrivial phases, respectively.

As shown in Fig. \ref{swap}(c), $|\phi_0\rangle$ and $|\phi_1\rangle$ can be exactly mapped to each other with a $2$-local fermionic circuit
\begin{equation}
C_\phi=\prod_{j=N-1}^1C^{(j)},~C^{(j)}=\frac{c_{2j+2}c_{2j+1}+c_{2j+1}c_{2j}}{\sqrt2}
\end{equation}
of depth $N-1$, where the local unitary $C^{(j)}$ swaps $c_{2j}$ and $c_{2j+2}$. As $|\psi_k\rangle$ and $|\phi_k\rangle$ are in the same phase, there exists a local fermionic circuit $C_k$ of constant depth (Appendix \ref{rig}) such that $|\langle\psi_k|P|\psi_k\rangle-\langle\psi'_k|P|\psi'_k\rangle|\le\epsilon$ for any local operator $P$ with bounded norm, where $|\psi'_k\rangle=C_k|\phi_k\rangle$. Finally, $C=C_1C_\phi C_0^\dag$ is the circuit of linear depth that connects $|\psi_0\rangle$ and $|\psi_1\rangle$.
\end{proof}

\begin{proposition} \label{prop2}
Suppose $|\psi\rangle$ and $C|\psi\rangle$ are two gapped ground states in 1D fermionic systems, where $C$ is a local fermionic circuit of sublinear depth. Then $|\psi\rangle$ and $C|\psi\rangle$ are in the same topological phase.
\end{proposition}

\begin{proof}
The string order parameter
\begin{equation} \label{strop}
\lim_{N\rightarrow+\infty}\left\langle\left(a_\frac{N}{3}^\dag+a_\frac{N}{3}\right)\prod_{j=\frac{N}{3}}^{\frac{2N}{3}-1}e^{i\pi a_j^\dag a_j}\left(a_\frac{2N}{3}^\dag+a_\frac{2N}{3}\right)\right\rangle
\end{equation}
is zero in the topologically trivial phase and nonzero in the topologically nontrivial phase \cite{BV13}. We show that its value cannot change between these two cases under local fermionic circuits of sublinear depth.

This is easiest to see by applying the Jordan-Wigner transformation
\begin{equation} \label{JW}
a_k=\sigma_k^-\prod_{j=1}^{k-1}(-\sigma_j^z),~a_k^\dag=\sigma_k^+\prod_{j=1}^{k-1}(-\sigma_j^z),
\end{equation}
where $\sigma_k^-$ and $\sigma_k^+$ are the spin-$1/2$ lowering and raising operators at the site $k$. The fermionic model (\ref{maj}) is mapped to the transverse field Ising model with periodic boundary conditions,
\begin{equation} \label{ising}
H=-\sum_{j=1}^{N-1}\sigma_j^x\sigma_{j+1}^x-\sigma_N^x\sigma_1^x+\mu\sum_{j=1}^N\sigma_j^z,
\end{equation}
and the string order parameter (\ref{strop}) is mapped to $\lim_{N\rightarrow+\infty}\langle\psi_s|\sigma_{N/3}^x\sigma_{2N/3}^x|\psi_s\rangle$, where $|\psi_s\rangle$ is the spin ground state. The spin model (\ref{ising}) is in the disordered phase for $\mu>1$ with vanishing correlations at large distances, e.g., $\lim_{N\rightarrow+\infty}\langle\psi_s|\sigma_{N/3}^x\sigma_{2N/3}^x|\psi_s\rangle=0$, and it is in the ordered phase for $0\le\mu<1$ with long-range correlations: $\lim_{N\rightarrow+\infty}\langle\psi_s|\sigma_{N/3}^x\sigma_{2N/3}^x|\psi_s\rangle>0$. As any local unitary in 1D fermionic systems remains local after the nonlocal Jordan-Wigner transformation (\ref{JW}) [in the case where the local unitary in 1D fermionic systems crosses the boundary, there is a trivial factor $\prod_{j=1}^N(-\sigma_j^z)=1$ as the fermion parity is even], a local fermionic circuit $C$ of sublinear depth is mapped to a local spin circuit $C_s$ of sublinear depth. The Lieb-Robinson bound states that correlations can only propagate at a finite speed in quantum many-body systems with local interactions \cite{LR72, NS06, HK06}. As a consequence, local quantum circuits of sublinear depth cannot generate long-range order \cite{BHV06}, i.e., $\lim_{N\rightarrow+\infty}\langle\psi_s|C_s^\dag\sigma_{N/3}^x\sigma_{2N/3}^xC_s|\psi_s\rangle=0$ for any state $|\psi_s\rangle$ with vanishing correlations at large distances. Therefore, the string order parameter (\ref{strop}) is either both zero or both nonzero for the fermionic states $|\psi\rangle$ and $C|\psi\rangle$.
\end{proof}

\section{Haldane chain} \label{sec4}

We switch to 1D spin systems. In the absence of symmetry, all 1D gapped spin systems are in the same phase. In the presence of symmetry, however, there can be SPT phases with degenerate edge states carrying projective representations of the symmetry group \cite{PTBO10, CGW11, CGW11a, SPC11}. See Appendix \ref{asec2} for the classification of 1D SPT phases, which includes a brief review of projective representations (Appendix \ref{proj_rep}). SPT states are short-range entangled in the sense that they can be created from product states with local quantum circuits of constant depth by breaking the symmetry. If the symmetry is preserved, we show that two ground states in different SPT phases can be connected by a local quantum circuit of linear depth and that linear depth is a lower bound.

We now study the Haldane chain with $Z_2\times Z_2$ on-site symmetry as a prototypical example, where we use periodic boundary conditions so that the ground state is unique and symmetric. The proof for general 1D SPT phases is similar but more complicated (Appendixes \ref{all_ld} and \ref{all_llb}). With $Z_2 \times Z_2$ symmetry, there are two phases \cite{PTBO10, PBTO12}: the trivial phase and the Haldane (nontrivial SPT) phase \cite{H83, H83a, AKLT87, AKLT88}.

\begin{proposition} \label{prop3}
Suppose $|\psi_0\rangle$ and $|\psi_1\rangle$ are two $Z_2\times Z_2$ symmetric gapped ground states in the trivial and the Haldane phases, respectively. Given an arbitrarily small constant $\epsilon$, there exist $|\psi'_0\rangle,|\psi'_1\rangle$ and a symmetric local quantum circuit $C$ of linear depth such that $|\psi'_1\rangle=C|\psi'_0\rangle$ and
\begin{equation}
|\langle\psi_k|P|\psi_k\rangle-\langle\psi'_k|P|\psi'_k\rangle|\le\epsilon~(k=0,1)
\end{equation}
for any local operator $P$ with bounded norm.
\end{proposition}

\begin{proof}
The proof proceeds analogously to that of Proposition \ref{prop1}. Figures \ref{swap}(a) and \ref{swap}(b) illustrate the structures of the RG fixed-point states $|\phi_0\rangle$ and $|\phi_1\rangle$ in the trivial and the Haldane phases, respectively, where each dot now represents a spin-$1/2$ degree of freedom transforming projectively under $\pi$ rotations about the $x,y,z$ axes. It is apparent that the edge state of $|\phi_1\rangle$ in the Haldane phase is twofold degenerate and transforms projectively while that of $|\phi_0\rangle$ in the trivial phase is trivial.

As shown in Fig. \ref{swap}(c), $|\phi_0\rangle$ and $|\phi_1\rangle$ can be exactly mapped to each other by applying $(N-1)$ $2$-local \textsc{swap} gates sequentially. These \textsc{swap} gates rearrange the singlets, are $Z_2\times Z_2$ symmetric and form a symmetric 2-local quantum circuit $C_\phi$ of depth $N-1$. As $|\psi_k\rangle$ and $|\phi_k\rangle$ are in the same SPT phase, there exists a symmetric local quantum circuit $C_k$ of constant depth (Appendix \ref{rig}) such that $|\langle\psi_k|P|\psi_k\rangle-\langle\psi'_k|P|\psi'_k\rangle|\le\epsilon$ for any local operator $P$ with bounded norm, where $|\psi'_k\rangle=C_k|\phi_k\rangle$. Finally, $C=C_1C_\phi C_0^\dag$ is the symmetric circuit of linear depth that connects $|\psi_0\rangle$ and $|\psi_1\rangle$.
\end{proof}

\begin{figure}
\includegraphics[width=\linewidth]{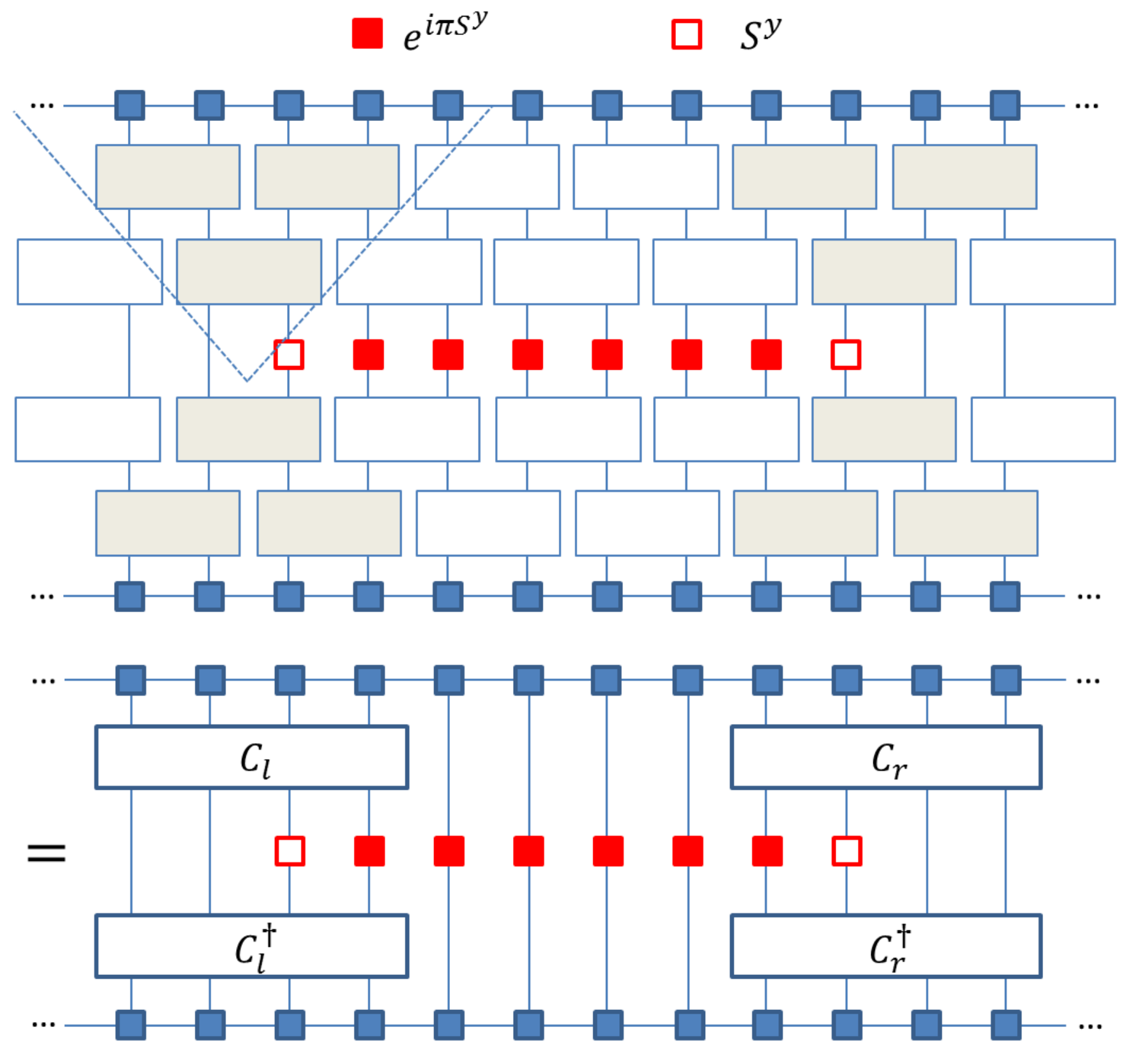}
\caption{(Color online) The expectation value $\langle\psi|C^\dag QC|\psi\rangle$. The horizontal lines attached with small blue squares represent $\langle\psi|$ (bra) or $|\psi\rangle$ (ket), and the short rectangles are the $2$-local unitaries in $C$. The (white) unitaries outside the causal cones (dotted lines) of $S^y$ (small open red squares) can be removed, as they are symmetric. Then we merge the (gray) symmetric local quantum gates inside each casual cone into one symmetric quantum gate (long rectangle) of sublinear support.} \label{sopm}
\end{figure}

\begin{proposition} \label{prop4}
Suppose $|\psi\rangle$ and $C|\psi\rangle$ are two symmetric gapped ground states in 1D spin systems with $Z_2\times Z_2$ on-site symmetry represented by $\left\{1,e^{i\pi S^x},e^{i\pi S^y},e^{i\pi S^z}\right\}$, where $C$ is a symmetric local quantum circuit of sublinear depth. Then $|\psi\rangle$ and $C|\psi\rangle$ are in the same SPT phase.
\end{proposition}

\begin{proof}
We make use of the string (nonlocal) order parameters \cite{HPCS12, PT12} distinguishing different SPT phases. For the Haldane chain, the string order operator is \cite{dNR89, KT92a, KT92b}
\begin{equation} \label{stroo}
Q=S_{N/3}^y\prod_{j=N/3+1}^{2N/3-1}e^{i\pi S_j^y}S_{2N/3}^y,
\end{equation}
where $\vec{S}_j=(S_j^x,S_j^y,S_j^z)$ is the spin-$1$ operator at the site $j$. The string order parameter $\lim_{N\rightarrow+\infty}\langle Q\rangle$ is zero in the trivial phase and nonzero in the Haldane phase. We show that its value cannot change between these two cases under $Z_2\times Z_2$ symmetric local quantum circuits of sublinear depth.

Assume without loss of generality that $C$ is a symmetric $2$-local quantum circuit of depth $m\le N/9$. Figure \ref{sopm} shows the expectation value $\langle \psi|C^{\dag}QC |\psi\rangle$. As each gate in the circuit $C$ is unitary and symmetric, the white gates cancel out. Then we merge the gray gates inside the causal cones (dotted lines) of the left and right end operators $S^y$ (small open red squares) into $C_l$ and $C_r$, respectively. As $C$ is of sublinear depth, $C_l$ and $C_r$ are nonoverlapping. Hence $Q'=C^{\dag}QC$ remains a string (order) operator. Specifically, the string becomes shorter but is still of the form $\prod_je^{i\pi S_j^y}$. The left and right end operators are changed to
\begin{eqnarray}
&&Q_l=C_l^\dag S_{N/3}^y\prod_{j=N/3+1}^{N/3+m}e^{i\pi S_j^y}C_l,\label{end1}\\
&&Q_r=C_r^\dag\prod_{j=2N/3-m}^{2N/3-1}e^{i\pi S_j^y}S_{2N/3}^yC_r,\label{end2}
\end{eqnarray}
respectively. As $C_l$ is symmetric, $Q_l$ transforms in the same way under the symmetry as $S^y$, e.g.,
\begin{equation}
\prod_je^{-i\pi S_j^z}Q_l\prod_je^{i\pi S_j^z}=-Q_l.
\end{equation}
Appendix \ref{full} shows that $\lim_{N\rightarrow+\infty}\langle\psi|Q'|\psi\rangle=0$ if and only if $\lim_{N\rightarrow+\infty}\langle\psi|Q|\psi\rangle=0$. Therefore, the string order operator (\ref{stroo}) has either both zero or both nonzero expectation values for $|\psi\rangle$ and $C|\psi\rangle$.
\end{proof}

Nonlocal (string) order parameters have been systematically constructed for general 1D SPT phases \cite{HPCS12, PT12}. Appendixes \ref{all_ld} and \ref{all_llb} extend our proof to all these cases accordingly.

\section{Conclusion} \label{sec5}

We have quantified the many-body entanglement in 1D (symmetry protected) topological states with (symmetric) local quantum circuits. In particular, we have shown that circuits of linear depth are necessary to generate 1D topological states from product states. We have also explicitly constructed circuits of linear depth that generate 1D topological states. These results are useful not only conceptually but also operationally as a guide to preparing topological states in experiments.

Although our proof is in 1D, we expect similar results in two and higher dimensions. Indeed, it has been shown that local quantum circuits of linear (in the diameter of the system) depth are necessary to generate states with topological degeneracy \cite{BHV06}. We conjecture that this is also true for topological states without topological degeneracy, e.g., the integer quantum Hall states, the $p$-wave superconductors, and the $E_8$ states. See Ref. \cite{Haa14} for recent progress in this direction.

More generally, we can ask, what is the quantum circuit complexity of generating ground states in gapless phases or at phase transitions? We expect that quantum circuits also characterize the entanglement patterns that give rise to the physical properties in gapless or critical systems.

\begin{acknowledgments}

We would like to thank Isaac H. Kim, Spyridon Michalakis, Joel E. Moore, John Preskill, Frank Pollmann, and Ashvin Vishwanath for helpful discussions. In particular, I.H.K. pointed out that a variant of Proposition \ref{prop2} can be proved using his entropic topological invariant \cite{Kim14}. This work was supported by the Miller Institute for Basic Research in Science at the University of California, Berkeley, the Caltech Institute for Quantum Information and Matter, the Walter Burke Institute for Theoretical Physics (X.C.), and DARPA OLE (Y.H.).

\end{acknowledgments}

\appendix

\section{STATES IN THE SAME PHASE---CONSTANT DEPTH} \label{rig}

We give a rigorous formulation of the statement \cite{CGW10} that gapped ground states in the same phase are connected by local quantum circuits of constant depth.

\begin{lemma} \label{Sch}
Suppose $H_0(t)$ and $H_1(t)$ are two time-dependent Hamiltonians with $\|H_0(t)-H_1(t)\|\le\delta$. Then the (unitary) time-evolution operators
\begin{equation}
U_k(t)=\mathcal Te^{-i\int_0^tH_k(\tau)d\tau}~(k=0,1)
\end{equation}
satisfy $\|U_0(t)-U_1(t)\|\le\delta t$, where $\mathcal T$ is the time-ordering operator.
\end{lemma}

\begin{proof}
Let
\begin{equation}
U_I(t)=\mathcal Te^{-i\int_0^tU_0^\dag(\tau)[H_1(\tau)-H_0(\tau)]U_0(\tau)d\tau}
\end{equation}
be the (unitary) time-evolution operator in the interaction picture. Indeed, it is straightforward to verify $U_1(t)=U_0(t)U_I(t)$ by differentiating with respect to $t$. Then,
\begin{eqnarray}
&&\|U'_I(t)\|=\|U_0^\dag(t)(H_1(t)-H_0(t))U_0(t)U_I(t)\|\nonumber\\
&&=\|H_1(t)-H_0(t)\|\le\delta\nonumber\\
&&\Rightarrow\|U_0(t)-U_1(t)\|=\|U_0(t)U_I(0)-U_0(t)U_I(t)\|\nonumber\\
&&=\|U_I(0)-U_I(t)\|\le\delta t.
\end{eqnarray}
\end{proof}

\begin{lemma} \label{LR}
Suppose $H(t)=\sum_{j=1}^{N-1}h^{(j)}(t)$ is a time-dependent 1D $2$-local Hamiltonian with open boundary conditions, where $h^{(j)}$ acts on the spins $j$ and $j+1$ (nearest-neighbor interaction). Define $H_*(t)=\sum_{j=1}^{l-1}h^{(j)}(t)$ for $l\le N$. Let $U(t)$ and $U_*(t)$ be the (unitary) time-evolution operators for $H(t)$ and $H_*(t)$, respectively. Then,
\begin{equation}
\|U^\dag(1)PU(1)-U_*^\dag(1)PU_*(1)\|=e^{-\Omega(l)}
\end{equation}
for any operator $P$ acting on the first spin with $\left\|P\right\|\le1$.
\end{lemma}

Lemma \ref{LR} is a variant of the Lieb-Robinson bound \cite{LR72, NS06, HK06}. See Ref. [24] in Ref. \cite{Osb06} for a simple direct proof.

\begin{theorem} [formal statement of Theorem \ref{thm1}] \label{Adi}
Suppose $|\psi_0\rangle$ and $|\psi_1\rangle$ are two gapped ground states in the same phase in any spatial dimension. Given an arbitrarily small constant $\epsilon=\Theta(1)$, there exists a local quantum circuit $C$ of depth $O(1)$ such that
\begin{equation}
|\langle\psi_1|P|\psi_1\rangle-\langle\psi_0|C^\dag PC|\psi_0\rangle|\le\epsilon
\end{equation}
for any local operator $P$ with $\|P\|\le1$.
\end{theorem}

\begin{proof}
By Definition \ref{def1}, there exists a smooth path of gapped local Hamiltonians $H_0(t)$ with $0\le t\le1$ such that $|\psi_0\rangle$ and $|\psi_1\rangle$ are the ground states of $H_0(0)$ and $H_0(1)$, respectively. Quasiadiabatic continuation \cite{HW05} defines a smooth time-dependent local Hamiltonian $H_1(t)$ such that
\begin{equation}
|\langle\psi_1|P|\psi_1\rangle-\langle\psi_0|U_1^\dag(1)PU_1(1)|\psi_0\rangle|\le\epsilon/3
\end{equation}
for any local operator $P$ with $\|P\|\le1$. Assume without loss of generality that $H_1(t)=\sum_{j=1}^{N-1}h_1^{(j)}(t)$ is a 1D $2$-local Hamiltonian with open boundary conditions and that $P$ is an operator acting on the first spin. We approximate the time-dependent Hamiltonian $H_1(t)$ by the piecewise time-independent Hamiltonian
\begin{equation}
\sum_{j=1}^{N-1}h_2^{(j)}=H_2(t):=H_1([rt]/r)=\sum_{j=1}^{N-1}h_1^{(j)}([rt]/r)
\end{equation}
with sufficiently large $r=O(1)$. Let $l=O(1)$ be a cutoff and define
\begin{equation}
H_3(t)=\sum_{j=1}^{l-1}h_1^{(j)}(t)+\sum_{j=l}^{N-1}h_2^{(j)}(t).
\end{equation}
Lemma \ref{LR} implies
\begin{equation}
\|U_1^\dag(1)PU_1(1)-U_3^\dag(1)PU_3(1)\|\le\epsilon/6
\end{equation}
for sufficiently large $l=O(1)$. As $H_1(t)$ is smooth, Lemma \ref{Sch} implies
\begin{eqnarray}
&&\lim_{r\rightarrow+\infty}\left\|h_1^{(j)}(t)-h_2^{(j)}(t)\right\|=0\nonumber\\
&&\Rightarrow\|H_3(t)-H_2(t)\|\le\sum_{j=1}^{l-1}\left\|h_1^{(j)}(t)-h_2^{(j)}(t)\right\|\le\epsilon/12\nonumber\\
&&\Rightarrow\|U_3(1)-U_2(1)\|\le\epsilon/12\nonumber\\
&&\Rightarrow\|U_3^\dag(1)PU_3(1)-U_2^\dag(1)PU_2(1)\|\le\epsilon/6
\end{eqnarray}
for sufficiently large $r=O(l/\epsilon)=O(1)$. Hence,
\begin{equation}
\|U_1^\dag(1)PU_1(1)-U_2^\dag(1)PU_2(1)\|\le\epsilon/3.
\end{equation}
As $H_2(t)$ is piecewise time independent, assume without loss of generality that it is time independent. Define
\begin{equation}
H_2=H^o+H^e,~H^o=\sum_{j=1}^{[N/2]}h_2^{(2j-1)},~H^e=\sum_{j=1}^{[(N-1)/2]}h_2^{(2j)}
\end{equation}
such that the first-order Trotter decomposition is given by
\begin{eqnarray}
&&U_2(1)=(e^{-iH^o/s-iH^e/s})^s\approx(e^{-iH^o/s}e^{-iH^e/s})^s\nonumber\\
&&=\left(\prod_{j=1}^{[N/2]}e^{-ih_2^{(2j-1)}/s}\prod_{j=1}^{[(N-1)/2]}e^{-ih_2^{(2j)}/s}\right)^s=:C,
\end{eqnarray}
where $C$ is a $2$-local quantum circuit of depth $2s$. Let $L=O(1)$ be a cutoff and define
\begin{eqnarray}
&&H_*=\sum_{j=1}^{L-1}h_2^{(j)}=H_*^o+H_*^e,\nonumber\\
&&H_*^o=\sum_{j=1}^{[L/2]}h_2^{(2j-1)},~H_*^e=\sum_{j=1}^{[(L-1)/2]}h_2^{(2j)}.
\end{eqnarray}
Similarly,
\begin{eqnarray}
&&U_*(1)=(e^{-iH_*^o/s-iH_*^e/s})^s\approx(e^{-iH_*^o/s}e^{-iH_*^e/s})^s\nonumber\\
&&=\left(\prod_{j=1}^{[L/2]}e^{-ih_2^{(2j-1)}/s}\prod_{j=1}^{[(L-1)/2]}e^{-ih_2^{(2j)}/s}\right)^s=:C_*,
\end{eqnarray}
where $C_*$ is also a $2$-local quantum circuit of depth $2s$. The standard error analysis of the Trotter decomposition leads to
\begin{eqnarray}
&&\|H_*\|=O(L)=O(1)\Rightarrow\|U_*(1)-C_*\|\le\epsilon/18\nonumber\\
&&\Rightarrow\|U_*^\dag(1)PU_*(1)-C_*^\dag PC_*\|\le\epsilon/9
\end{eqnarray}
for sufficiently large $s=O(1)$. We observe that $C=\mathcal Te^{-i\int_0^2H^C(t)dt}$ is the (unitary) time-evolution operator for the piecewise time-independent Hamiltonian $H^C(t)$, where $H^C(t)=H^o$ if $[st]$ is odd and $H^C(t)=H^e$ if $[st]$ is even. Similarly, $C_*=\mathcal Te^{-i\int_0^2H_*^C(t)dt}$, where $H_*^C(t)=H_*^o$ if $[st]$ is odd and $H_*^C(t)=H_*^e$ if $[st]$ is even. Lemma \ref{LR} implies
\begin{eqnarray}
&&\|U_2^{\dag}(1)PU_2(1)-U_*^\dag(1)PU_*(1)\|\le\epsilon/9,\\
&&\|C^\dag PC-C_*^\dag PC_*\|\le\epsilon/9
\end{eqnarray}
for sufficiently large $L=O(1)$. Hence,
\begin{equation}
\|U_2^\dag(1)PU_2(1)-C^\dag PC\|\le\epsilon/3.
\end{equation}
Finally,
\begin{eqnarray}
&&|\langle\psi_1|P|\psi_1\rangle-\langle\psi_0|C^\dag PC|\psi_0\rangle|\nonumber\\
&&\le|\langle\psi_1|P|\psi_1\rangle-\langle\psi_0|U_1^\dag(1)PU_1(1)|\psi_0\rangle|\nonumber\\
&&+|\langle\psi_0|U_1^\dag(1)PU_1(1)|\psi_0\rangle-\langle\psi_0|C^\dag PC|\psi_0\rangle|\nonumber\\
&&\le\epsilon/3+\|U_1^\dag(1)PU_1(1)-C^\dag PC\|\nonumber\\
&&\le\epsilon/3+\|U_1^\dag(1)PU_1(1)-U_2^\dag(1)PU_2(1)\|\nonumber\\
&&+\|U_2^\dag(1)PU_2(1)-C^\dag PC\|\nonumber\\
&&\le\epsilon/3+\epsilon/3+\epsilon/3=\epsilon.
\end{eqnarray}
\end{proof}

A minor modification of the proof of Theorem \ref{Adi} leads to similar results in fermionic systems and/or in the presence of symmetry.

\begin{corollary} [formal statement of Corollary \ref{cor1}] \label{coro}
Suppose $|\psi_0\rangle$ and $|\psi_1\rangle$ are two symmetric gapped ground states in the same SPT phase in any spatial dimension. Given an arbitrarily small constant $\epsilon=\Theta(1)$, there exists a symmetric local quantum circuit $C$ of depth $O(1)$ such that
\begin{equation}
|\langle\psi_1|P|\psi_1\rangle-\langle\psi_0|C^\dag PC|\psi_0\rangle|\le\epsilon
\end{equation}
for any local operator $P$ with $\|P\|\le1$.
\end{corollary}

The main result of Ref. \cite{Osb07} is an immediate corollary of Theorem \ref{Adi}.

\begin{corollary} [efficient classical simulation of adiabatic quantum computation with a constant gap in any spatial dimension]
Suppose we are given a smooth path of gapped local Hamiltonians $H(t)$ with $0\le t\le1$, where the ground state $|\psi_0\rangle$ of $H(0)$ is simple in the sense that $\langle\psi_0|P|\psi_0\rangle$ can be efficiently computed classically for any local operator $P$ with $\|P\|\le1$. Then $\langle\psi_1|P|\psi_1\rangle$ can be efficiently computed classically up to an arbitrarily small constant additive error, where $|\psi_1\rangle$ is the ground state of $H(1)$ encoding the solution of the adiabatic quantum computation.
\end{corollary}

\section{SYMMETRY PROTECTED TOPOLOGICAL PHASE} \label{asec2}

We review the classification of 1D SPT phases (Appendix \ref{sptc}), and begin by recalling two key notions: projective representations (Appendix \ref{proj_rep}) and matrix product states (Appendix \ref{sptb}).

\subsection{Projective representation} \label{proj_rep}

In the context of this paper, a projective representation is a mapping $u$ from the symmetry group $G$ to unitary matrices such that
\begin{equation}
u(g_1)u(g_2)=\omega(g_1,g_2)u(g_1g_2),
\end{equation}
where $\omega(g_1,g_2)$ (called the factor system of the projective representation) is a $U(1)$ phase factor, cf. $u$ is a linear representation of $G$ if the factor system is trivial, i.e., $\omega(g_1,g_2)=1$ for any $g_1,g_2\in G$. The associativity of $G$ implies
\begin{equation} \label{2cocycle_om}
\omega(g_2,g_3)\omega(g_1,g_2g_3)=\omega(g_1,g_2)\omega(g_1g_2,g_3).
\end{equation}
Multiplying $u$ by $U(1)$ phase factors leads to a different projective representation $u'$ with the factor system $\omega'$:
\begin{eqnarray} \label{omom}
&&u'(g)=\beta(g)u(g),~\beta(g)\in U(1),~\forall g\in G\nonumber\\
&&\Rightarrow\omega'(g_1,g_2)=\omega(g_1,g_2)\beta(g_1)\beta(g_2)/\beta(g_1g_2).
\end{eqnarray}
Two projective representations $u$ and $u'$ are equivalent if and only if they differ only by prefactors. Correspondingly, their factor systems $\omega$ and $\omega'$ are said to be in the same equivalence class $[\omega]$. Let $u_1$ and $u_2$ be two projective representations with the factor systems $\omega_1$ and $\omega_2$ in the equivalence classes $[\omega_1]$ and $[\omega_2]$, respectively. Apparently, $u_1\otimes u_2$ is a projective presentation with the factor system $\omega_1\omega_2$ in the equivalence class $[\omega_1\omega_2]$. By defining $[\omega_1]\cdot[\omega_2]=[\omega_1\omega_2]$, the equivalence classes of factor systems form an Abelian group [called the second cohomology group $H^2(G,U(1))$], where the identity element is the equivalence class that contains the trivial factor system.

\subsection{Matrix product state} \label{sptb}

Suppose we are working with a chain of $N$ spins (qudits), and the local dimension of each spin is $d=\Theta(1)$. Let $\{|i_k\rangle\}_{i_k=1}^d$ be the computational basis of the Hilbert space of the spin $k$.

\begin{figure}
\includegraphics[width=\linewidth]{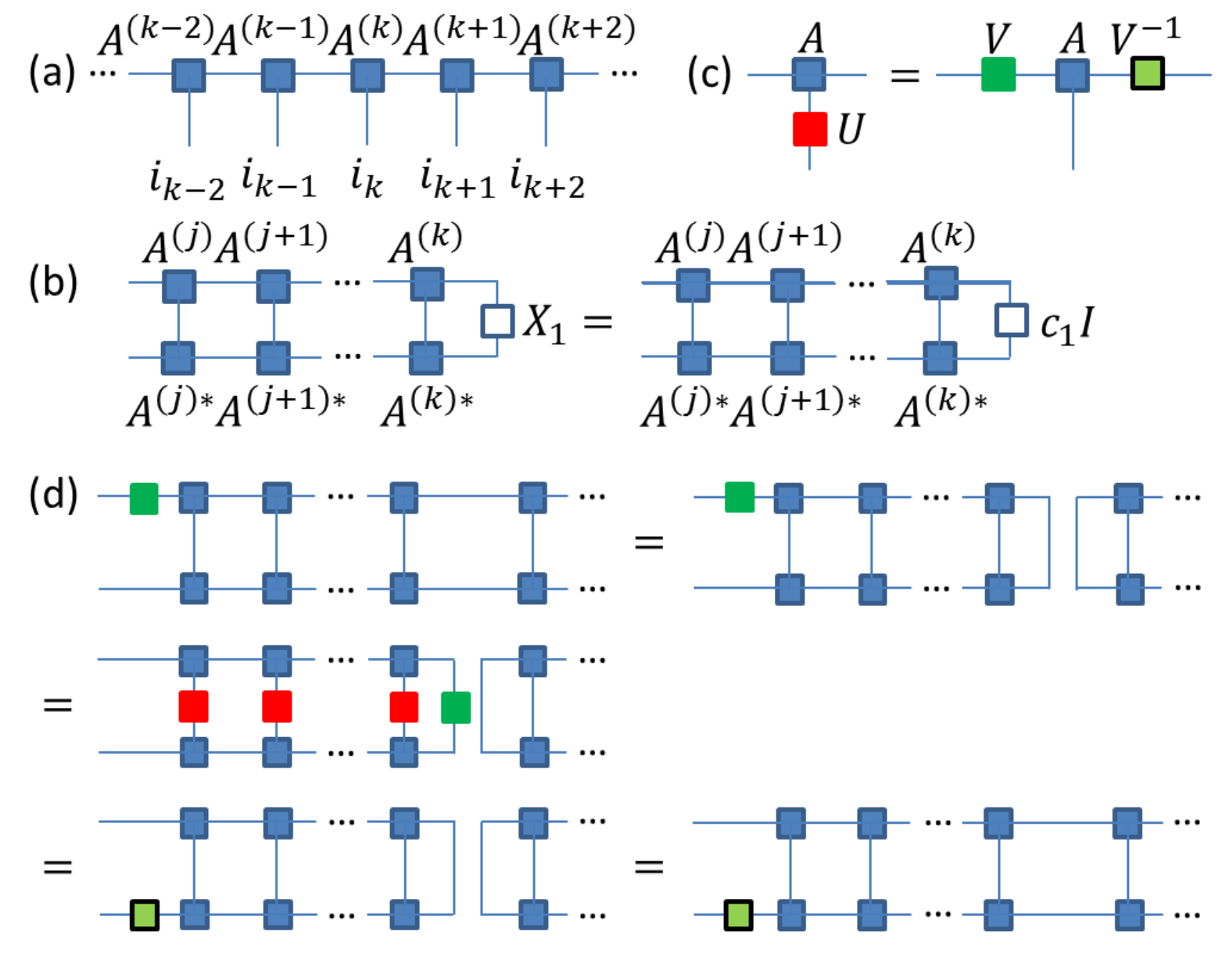}
\caption{(Color online) (a) Graphical representation of MPS (\ref{MPS}) \cite{Sch11}. Each square represents a tensor $A^{(k)}$ with two bond indices (horizontal lines) and one physical index (vertical line). The bond indices are contracted sequentially with periodic boundary conditions (not shown). (b) The condition (\ref{SRC1}) for short-range correlated MPS. The graphical equation is approximate up to error $e^{-\Omega(k-j)}$, which can be neglected in the thermodynamic limit $N\rightarrow+\infty$ if $k-j=\Theta(N)$. (c) Graphical representation of (\ref{sym}). The site labels are not shown. (d) is a consequence of (b) and (c). Note that a prefactor of the second, third, and fourth tensor networks is not shown.} \label{Pre}
\end{figure}

\begin{definition} [matrix product state (MPS) \cite{PVWC07, FNW92}]
Let $\{D_k\}_{k=0}^n$ with $D_0=D_n$ be a sequence of positive integers. As illustrated in Fig. \ref{Pre}(a), an MPS $|\Psi\rangle$ takes the form
\begin{equation} \label{MPS}
|\Psi\rangle=\sum_{i_1,i_2,\ldots,i_N=1}^d\mathrm{tr}\left(A_{i_1}^{(1)}A_{i_2}^{(2)}\cdots A_{i_N}^{(N)}\right)|i_1i_2\cdots i_N\rangle,
\end{equation}
where $A_{i_k}^{(k)}$ is a matrix of size $D_{k-1}\times D_k$. Define $D=\max\{D_k\}_{k=0}^n$ as the bond dimension of the MPS $|\Psi\rangle$.
\end{definition}

The ground states of 1D gapped Hamiltonians can be represented as MPSs of small bond dimension \cite{Has07, VC06}. The ground states of gapped local Hamiltonians are short-range correlated in the sense that all connected correlation functions decay exponentially with distance \cite{Has04, NS06, HK06}.

For each $k$, define two linear maps
\begin{equation}
\mathcal E_k(X)=\sum_{i_k=1}^dA_{i_k}^{(k)}XA_{i_k}^{(k)\dag},~\mathcal E_k^*(X)=\sum_{i_k=1}^dA_{i_k}^{(k)\dag}XA_{i_k}^{(k)}.
\end{equation}
Any MPS can be transformed into the so-called canonical form \cite{PVWC07} such that $\mathcal E_k(I)=I$ and $\mathcal E_k^*(M_{k-1})=M_k$, where $I$ is an identity matrix, and $M_k$ is a positive diagonal matrix. A canonical MPS is short-range correlated if for any $X_1,X_2$ with $\|X_1\|,\|X_2\|\le1$ there exist coefficients $c_1,c_2$ such that
\begin{eqnarray} 
&&\|\mathcal E_j\mathcal E_{j+1}\cdots\mathcal E_k(X_1-c_1I)\|=e^{-\Omega(k-j)},\label{SRC1}\\
&&\|\mathcal E_k^*\mathcal E_{k-1}^*\cdots\mathcal E_j^*(X_2-c_2M_{j-1})\|=e^{-\Omega(k-j)}\label{SRC2}
\end{eqnarray}
at large $k-j$, i.e., $X_1$ can be replaced by $c_1I$ up to error $e^{-\Omega(k-j)}$, as illustrated in Fig. \ref{Pre}(b). Hence $X_1$ (and $X_2$) can be replaced by any matrix up to a multiplicative prefactor and an exponentially small error. When $A_{i_k}^{(k)}$'s are site independent (and the MPS $|\Psi\rangle$ is translationally invariant), (\ref{SRC1}) and (\ref{SRC2}) are equivalent to the condition \cite{FNW92, PVWC07} that the second largest (in magnitude) eigenvalue $|\nu_2|$ of $\mathcal E_k$ is less than $1$, and the left-hand sides of (\ref{SRC1}) and (\ref{SRC2}) decay as $O(|\nu_2|^{-(k-j)})$.

\subsection{Classification of 1D SPT phases} \label{sptc}

1D SPT phases are completely characterized by the degenerate edge states carrying projective representations of the symmetry group, i.e., there is a one-to-one correspondence between 1D SPT phases and the equivalence classes of projective representations. The edge states can be easily seen from the short-range correlated MPS representation (\ref{MPS}) of SPT states. Suppose $U$ is an on-site symmetry with the symmetry group $G$, i.e., $U$ is an isomorphism of $G$ such that $U(g)^{\otimes N}|\Psi\rangle=|\Psi\rangle$ for any $g\in G$. Recall that $\{|i_k\rangle\}_{i_k=1}^d$ is the computational basis of the Hilbert space of the spin $k$. One can show that $A_{i_k}^{(k)}$'s satisfy \cite{PWS+08, CGW11}
\begin{equation} \label{sym}
\sum_{i'_k}\langle i_k|U(g)|i'_k\rangle A_{i'_k}^{(k)}=e^{i\theta(g)}V_{k-1}(g)A_{i_k}^{(k)}V_k^{-1}(g),
\end{equation}
as illustrated in Fig. \ref{Pre}(c). Furthermore, $e^{i\theta(g)}$ is a 1D representation of $G$. It can be effectively eliminated by blocking sites unless $G$ has an infinite number of 1D representations \cite{CGW11}; here we drop $e^{i\theta(g)}$ for simplicity. $V_k(g)$ is a projective representation of $G$. The equivalence class of $V_k(g)$ is site independent and labels the SPT phase of the MPS $|\Psi\rangle$. As such, 1D SPT phases are classified by the second cohomology group $H^2(G,U(1))$ in the presence of an on-site symmetry $U$ \cite{CGW11, SPC11}. In particular, all 1D gapped spin systems are in the same phase in the absence of symmetry \cite{CGW11, SPC11}, cf. $H^2(G,U(1))$ is trivial if $G$ is trivial.

1D SPT phases can be detected by nonlocal (string) order parameters. When the symmetry group $G$ is Abelian, there is a set of string order parameters from which the SPT phase of any symmetric gapped ground state can be extracted \cite{PT12, Mar13}. When $G$ is not necessarily Abelian, a different and more complicated type of nonlocal order parameter fully characterizes SPT phases \cite{HPCS12, PT12}.

\section{COMPLETE PROOF OF PROPOSITION \ref{prop4}} \label{full}

\begin{figure}
\includegraphics[width=\linewidth]{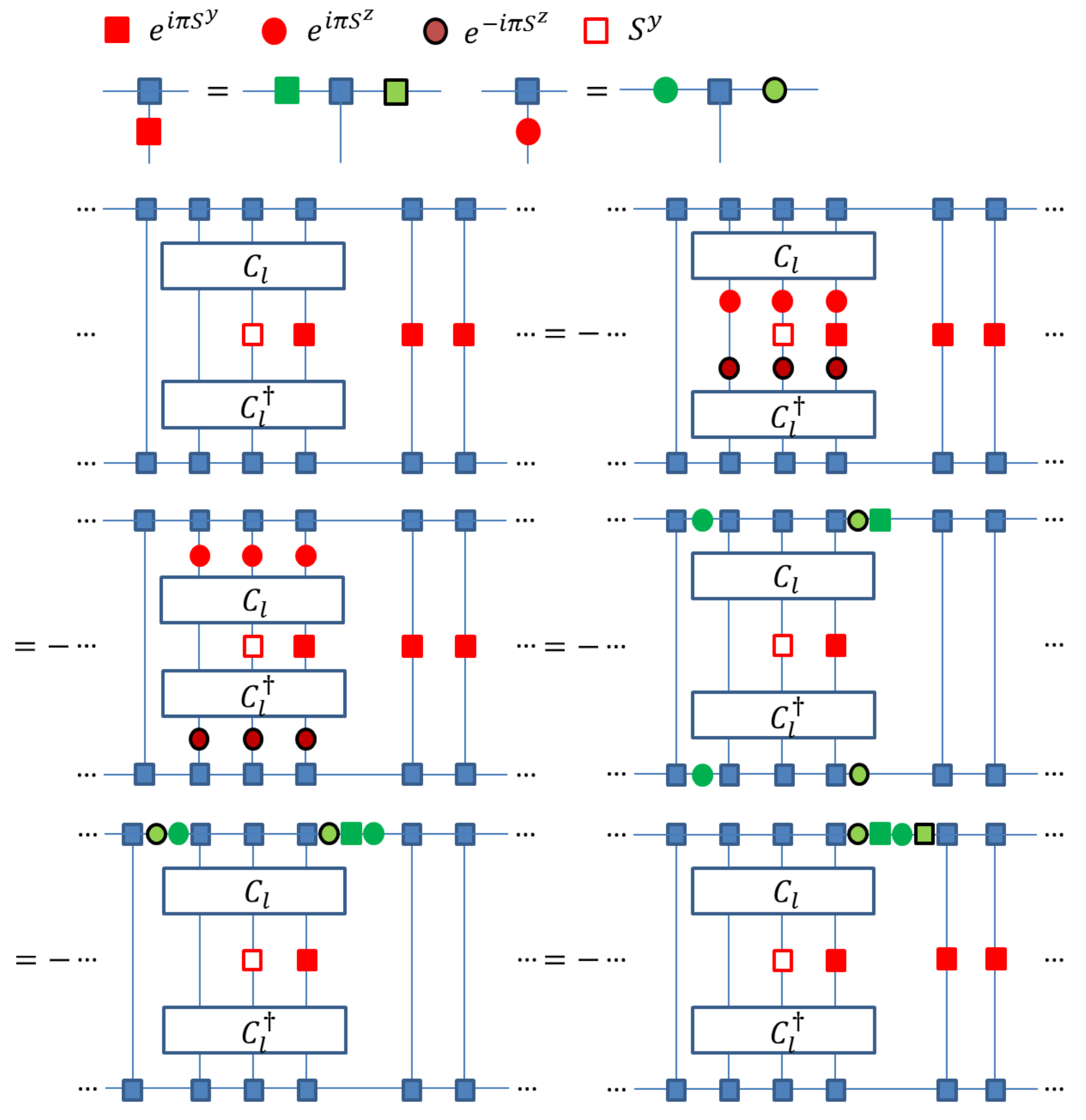}
\caption{(Color online) Graphical proof of $\langle\psi|Q'|\psi\rangle=0$ in the thermodynamic limit $N\rightarrow+\infty$ under the assumption that $|\psi\rangle$ is in the trivial phase.} \label{sop}
\end{figure}

\begin{proof} [Proof of Proposition \ref{prop4}]
We use the string order operator $Q$ (\ref{stroo}). Its expectation value $\lim_{N\rightarrow+\infty}\langle Q\rangle$ is zero in the trivial phase and nonzero in the Haldane phase. As shown in Fig. \ref{sopm}, $Q'=C^\dag QC=Q_l\prod_{j=N/3+m+1}^{2N/3-m-1}e^{i\pi S^y_j}Q_r$ remains a string (order) operator, where the end operators $Q_l$ and $Q_r$ are given by (\ref{end1}) and (\ref{end2}), respectively. It suffices to prove $\lim_{N\rightarrow+\infty}\langle\psi|Q'|\psi\rangle=0$ under the assumption that $|\psi\rangle$ is in the trivial phase.

See Fig. \ref{sop} for a graphical proof.  We focus on the left end of the string (order) operator $Q'$. The green squares and circles carry projective representations induced by the corresponding symmetry operators (red squares and circles, respectively) [cf. Fig. \ref{Pre}(c)]. We briefly explain each step of the graphical equation chain in Fig. \ref{sop}:\\
\indent Step 1: $e^{-i\pi S^z}S^ye^{i\pi S^z}=-S^y$ and $e^{-i\pi S^z}S^ze^{i\pi S^z}=S^z$.\\
\indent Step 2: $C_l$ is symmetric.\\
\indent Step 3: (\ref{sym}) Figure \ref{Pre}(c).\\
\indent Step 4: Figure \ref{Pre}(d).\\
\indent Step 5: (\ref{sym}) Figure \ref{Pre}(c).

In the last tensor network, the four green objects together contribute a trivial phase factor as $|\psi\rangle$ is in the trivial phase. Therefore, the first tensor network is zero due to the minus signs in the graphical equation chain.
\end{proof}

\section{STATES IN DIFFERENT PHASES---LINEAR DEPTH} \label{all_ld}

\begin{theorem}
Suppose $|\psi_0\rangle$ and $|\psi_1\rangle$ are two symmetric gapped ground states in different SPT phases. Given an arbitrarily small constant $\epsilon=\Theta(1)$, there exist $|\psi'_0\rangle,|\psi'_1\rangle$ and a symmetric local quantum circuit $C$ of depth $O(N)$ such that $|\psi'_1\rangle=C|\psi'_0\rangle$ and
\begin{equation}
|\langle\psi_k|P|\psi_k\rangle-\langle\psi'_k|P|\psi'_k\rangle|\le\epsilon~(k=0,1)
\end{equation}
for any local operator $P$ with $\|P\|\le1$.
\end{theorem}

\begin{proof}
The proof proceeds analogously to that of Proposition \ref{prop3}. Assume without loss of generality that $|\psi_k\rangle$ is in a nontrivial SPT phase. Let $|\phi\rangle$ be the RG fixed-point state in the trivial SPT phase, and $|\phi_k\rangle$ be the RG fixed-point state in the same SPT phase as $|\psi_k\rangle$. Figures \ref{swap}(a) and \ref{swap}(b) illustrate the structures of $|\phi\rangle$ and $|\phi_k\rangle$, respectively.

As shown in Fig. \ref{swap}(c), $|\phi\rangle$ and $|\phi_k\rangle$ can be exactly mapped to each other by applying $O(N)$ $2$-local \textsc{swap} gates sequentially. These \textsc{swap} gates are symmetric with respect to any on-site symmetry and form a symmetric $2$-local quantum circuit $C_{\phi,k}$ of depth $O(N)$. As $|\psi_k\rangle$ and $|\phi_k\rangle$ are in the same SPT phase, there exists a symmetric local quantum circuit $C_k$ of depth $O(1)$ (Corollary \ref{coro}) such that $|\langle\psi_k|P|\psi_k\rangle-\langle\psi'_k|P|\psi'_k\rangle|\le\epsilon$ for any local operator $P$ with $\|P\|\le1$, where $|\psi'_k\rangle=C_k|\phi_k\rangle$. Finally, $C=C_1C_{\phi,1}C_{\phi,0}^\dag C_0^\dag$ is the symmetric circuit of linear depth that connects $|\psi_0\rangle$ and $|\psi_1\rangle$.
\end{proof}

\section{STATES IN DIFFERENT PHASES---LINEAR LOWER BOUND} \label{all_llb}

The proof of Proposition \ref{prop4} can be generalized to other Abelian on-site symmetry. Indeed, string order parameters do (do not) fully characterize 1D SPT phases with Abelian (non-Abelian) on-site symmetry \cite{PT12, Mar13}. When the symmetry group is not necessarily Abelian, a different and more complicated type of nonlocal order parameter \cite{HPCS12,PT12} measures all gauge-invariant phase factors, which provide a complete description of the equivalence class of projective representations.

\begin{figure*}
\includegraphics[width=0.534\linewidth]{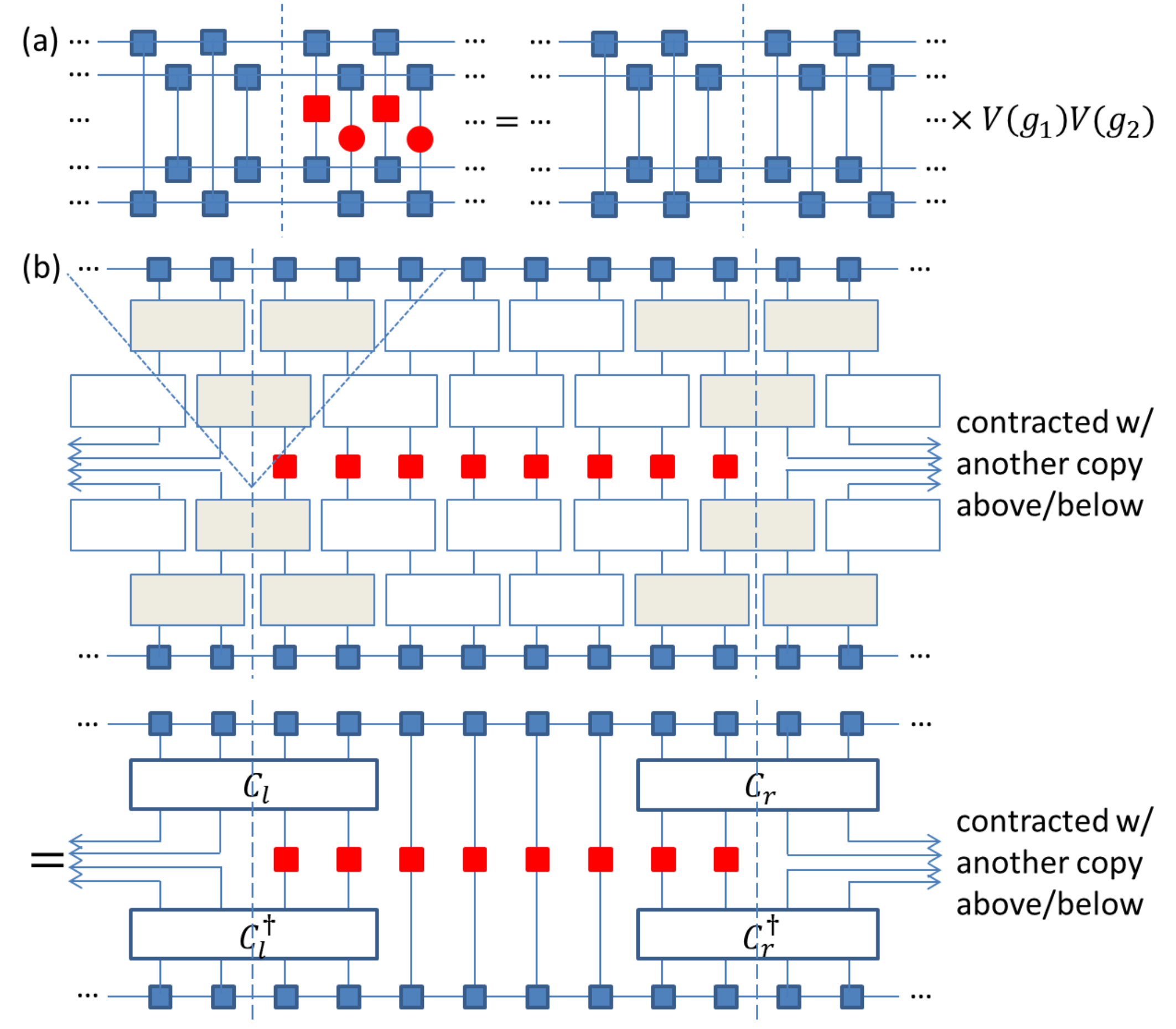}
\includegraphics[width=0.46\linewidth]{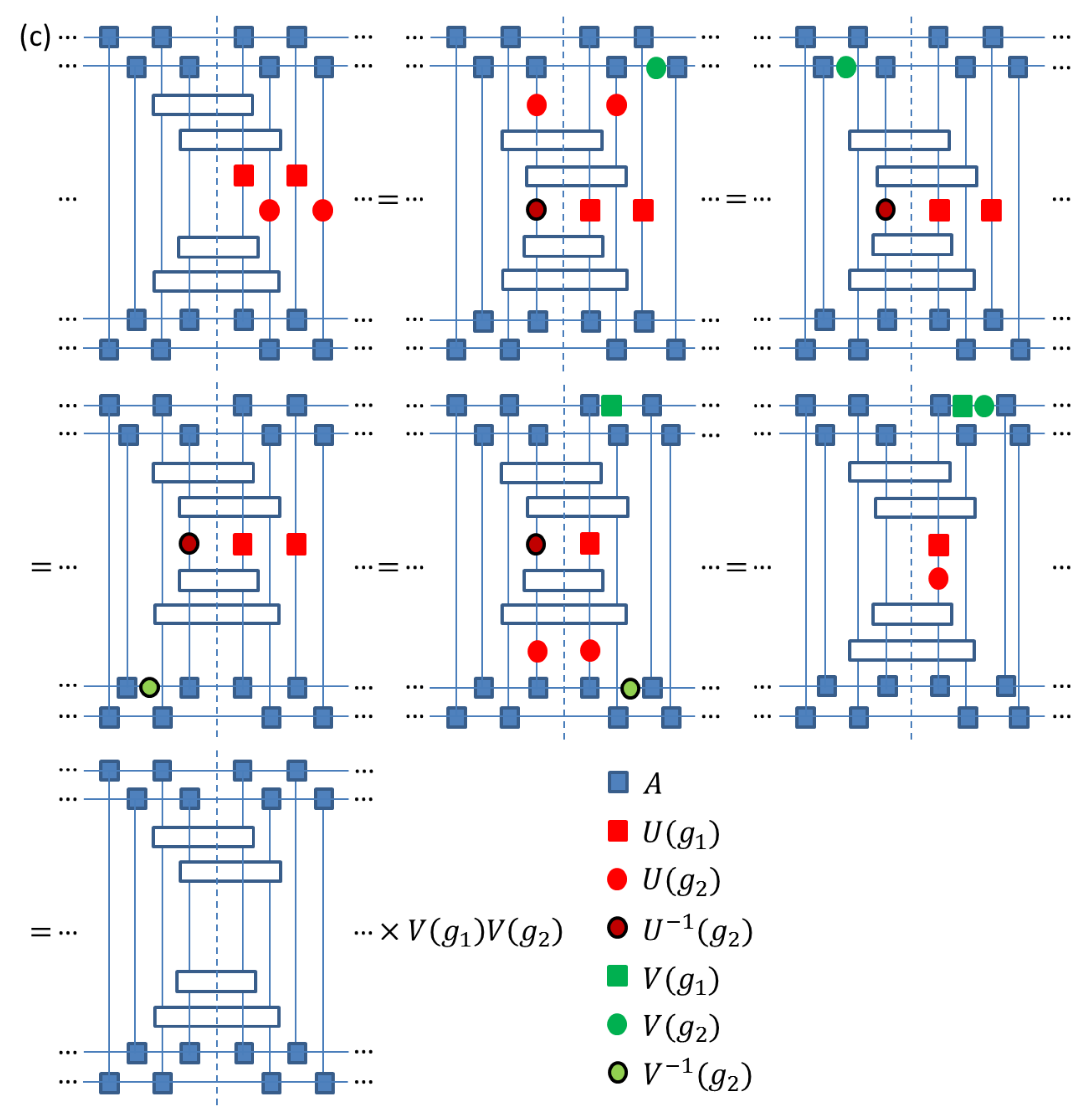}
\caption{(Color online) (a) The domain wall (dashed line) that contributes the local phase factor $V(g_1)V(g_2)$ \cite{PT12}. (b) The short rectangles are the local unitaries in $C$. The (white) unitaries outside the causal cones (dotted lines) of the domain walls can be removed, as they are symmetric. Then we merge the (gray) symmetric local quantum gates inside each casual cone into one symmetric quantum gate (long rectangle) of sublinear support. (c) Graphical proof of the invariance of the local phase factor for the domain wall in (a) under symmetric local quantum circuits of sublinear depth.} \label{Net}
\end{figure*}

\begin{theorem} \label{Gen}
Suppose $|\psi\rangle$ and $C|\psi\rangle$ are two symmetric gapped ground states in 1D spin systems with an on-site symmetry $U$, where $C$ is a symmetric local quantum circuit of sublinear depth. Then $|\psi\rangle$ and $C|\psi\rangle$ are in the same SPT phase.
\end{theorem}

\begin{proof}
As gauge-invariant phase factors provide a complete description of the equivalence class of projective representations, it suffices to show that all gauge-invariant phase factors cannot change under symmetric local quantum circuits of sublinear depth. Let $V$ be the projective representation of the symmetry group $G$ that labels the SPT phase of $|\psi\rangle$. The simplest example of a gauge-invariant phase factor is $V(g_1)V(g_2)V^{-1}(g_1)V^{-1}(g_2)$ for $g_1,g_2\in G$ with $U(g_1)U(g_2)U^{-1}(g_1)U^{-1}(g_2)=1$. However, the graphical representation of the nonlocal order parameter that measures this gauge-invariant phase factor contains eight copies of $|\psi\rangle$ (see Fig. 9 in Ref. \cite{PT12}) and is cumbersome. In order to simplify the illustration of our proof, we pretend that $V(g_1)V(g_2)$  with $U(g_1)U(g_2)=1$ is a gauge-invariant phase factor so that the corresponding nonlocal order parameter contains only four copies of $|\psi\rangle$. We show that this ``gauge-invariant phase factor'' cannot change under symmetric local quantum circuits of sublinear depth. It is straightforward to generalize the proof to any gauge-invariant phase factor.

We briefly review the construction of the tensor network (nonlocal order parameter) that measures the gauge-invariant phase factor $V(g_1)V(g_2)$ (see Sec. IV B of Ref. \cite{PT12} for details). The tensor network contains three domain walls (two of which are illustrated in Fig. 9 of Ref. \cite{PT12}). As $|\psi\rangle$ is short-range correlated in the sense of (\ref{SRC1}) and (\ref{SRC2}), one can define a ``local phase factor'' for each domain wall such that the overall phase factor is the product of all three local phase factors. Specifically, the domain wall in Fig. \ref{Net}(a) (corresponding to the left domain wall in Fig. 9 of Ref. \cite{PT12}) contributes the local phase factor $V(g_1)V(g_2)$. The other two domain walls (not shown) are $\Theta(N)$ sites away; they do not contribute any nontrivial local phase factors, but are necessary for restoring periodic boundary conditions. The left-hand side of the graphical equation in Fig. \ref{Net}(a) is constructed as follows. We take four copies of $|\psi\rangle$ (expressed as MPS): two copies above and two copies below [tensors in the copies below are complex conjugated as in Fig. \ref{Pre}(b)]; contract them via a permutation to the left and via the symmetry operators $U(g_1),U(g_2)$ (red squares and circles) to the right of the domain wall. Then the local phase factor $V(g_1)V(g_2)$ pops out, as illustrated in Fig. \ref{Net}(a).

Under symmetric local quantum circuits of sublinear depth, Fig. \ref{Net}(b) shows that the local phase factor for each domain wall is still well defined and Fig. \ref{Net}(c) proves its invariance. Specifically, in Fig. \ref{Net}(c) we assume without loss of generality that $C$ is a symmetric $2$-local quantum circuit of depth $1$ so that all four rectangles [corresponding to the gates $C_l$ and $C_l^\dag$ in Fig. \ref{Net}(b)] in each tensor network are symmetric and $2$-local. The first (from above to below) rectangle acts on the third and fifth (from left to right) vertical lines; the second acts on the fourth and sixth; the third acts on the fourth and fifth; the fourth acts on the third and sixth. All other crossings between rectangles and vertical lines should not be there if we could draw the tensor networks in 3D rather than in 2D. We briefly explain each step of the graphical equation chain in Fig. \ref{Net}(c):\\
\indent Step 1: (\ref{sym}) Figure \ref{Pre}(c) and the symmetry of the rectangles.\\
\indent Step 2: (\ref{sym}) Figure \ref{Pre}(c).\\
\indent Step 3: Figure \ref{Pre}(d).\\
\indent Step 4: (\ref{sym}) Figure \ref{Pre}(c).\\
\indent Step 5: Figure \ref{Pre}(d) and the symmetry of the rectangles.\\
\indent Step 6: $U(g_1)U(g_2)=1$.
\end{proof}

\begin{remark}
The time-reversal symmetry is not an on-site symmetry as the antiunitary time-reversal operator cannot be expressed as a tensor product of on-site operators. However, it can be effectively treated as an on-site symmetry using the trick in Sec. IV B of Ref. \cite{PT12}. Therefore, we expect that the proof of Theorem \ref{Gen} can be generalized to the time-reversal symmetry.
\end{remark}

\end{document}